\documentclass{article}
\usepackage{amsmath,amsfonts,amssymb,amsthm}
\usepackage{paralist}
\usepackage{booktabs}
\usepackage[round]{natbib}
\usepackage{multirow}
\usepackage{kbordermatrix}

\allowdisplaybreaks

\newcommand{\cd}{{\cal D}}

\newcommand{\bsw}{\boldsymbol{w}}
\newcommand{\bsx}{\boldsymbol{x}}
\newcommand{\bsy}{\boldsymbol{y}}
\newcommand{\bsz}{\boldsymbol{z}}

\newcommand{\glu}{{\!\::\:\!}}

\newcommand{\rd}{\mathrm{\, d}}

\newcommand{\var}{{\mathrm{Var}}}

\newcommand{\wh}{\widehat}

\renewcommand{\emptyset}{\varnothing}
\renewcommand{\ge}{\geqslant}
\renewcommand{\le}{\leqslant}

\newcommand{\dustd}{\mathbf{U}} 

\newcommand{\e}{{\mathbb{E}}} 

\newtheorem{theorem}{Theorem}
\newtheorem{proposition}{Proposition}

\theoremstyle{definition}

\renewcommand{\glu}{{:}}

\newcommand{\ult}{\underline\tau}
\newcommand{\olt}{\overline\tau}
\newcommand{\simiid}{\stackrel{\mathrm{iid}}{\sim}}

\author{
Art B. Owen\\
Stanford University}
\date{April 2012}
\title{Better estimation of small Sobol' sensitivity indices}
\begin{document}

\maketitle
\begin{abstract}
A new method for estimating Sobol' indices is proposed.
The new method makes use of $3$ independent input
vectors rather than the usual $2$.
It attains much greater accuracy on problems where
the target Sobol' index is small, even outperforming
some oracles which adjust using the true but unknown
mean of the function. When the target Sobol' index
is quite large, the oracles do better than the new method.
\end{abstract}

\section{Introduction}

Let $f$ be a deterministic function on $[0,1]^d$
for $d\ge 1$. 
Sobol' sensitivity indices, derived from
a functional ANOVA,
are used to measure
the importance of subsets of input variables.
There are two main types of index, but one
of them is especially hard to estimate
in cases where that index is small.

The problematic index can be represented
as a covariance between outcomes of $f$ evaluated at
two random input points, that share some but
not all of their components. A natural
estimator then is a sample covariance
based on pairs of random $d$-vectors
of this type.
\cite{sobo:mysh:2007} report a numerical
experiment where enormous efficiency differences
obtain depending on how one estimates that covariance.
The best gains
arise from applying some centering strategies
to those pairs of function evaluations.

This article introduces a new estimator
for the Sobol' index, based on three
input vectors, not two.  The new estimator
makes perhaps surprising use of randomly
generated centers. The random centering adds
to the cost of every simulation run and
might be thought to add noise. But in
many examples that noise must be strongly negatively
correlated with the quantity it adjusts
because (in those examples) the random centering greatly increases
efficiency. The new estimate is not always
most efficient.  In particular when the
index to be estimated is large the new
estimate is seen to perform worse than some oracles
that one could approximate numerically.

The motivation behind
Sobol' indices, is well explained in the 
text by~\cite{salt:ratt:andr:camp:cari:gate:sais:tara:2008}.
These indices have been applied to problems in
industry, science and public health.
For a recent mathematical account of Sobol' indices, 
see~\cite{sobomat:tr}. 

The outline of this article is as follows.
Section~\ref{sec:back} introduces
Sobol' indices and our notation.
Section~\ref{sec:esti} presents
the original estimator of the Sobol'
indices and the four improved estimators we consider here.
Section~\ref{sec:exam} considers some numerical
examples. For small Sobol' indices, the newly
proposed estimator is best, beating two oracles.
For very large indices, the best performance comes
from an oracle that uses the true function mean
twice.
Section~\ref{sec:gene} presents some theoretical
support for the new estimator.  It generalizes
the estimator to a wider class of methods
and shows that the proposed estimator minimizes
a proxy for the variance, when one considers functions
of product form.

\section{Background}\label{sec:back}

For $d\ge1$, let $f\in L^2[0,1]^d$.
Then $f$ can be written in an ANOVA decomposition
as a sum of a constant $\mu = \int_{[0,1]^d} f(\bsx)\rd\bsx$
and $2^d-1$ mutually orthogonal ANOVA effects, one for each nonempty
subset of $\cd=\{1,2,\dots,d\}$.
The effect for non-empty subset $u\subseteq\cd$
has variance $\sigma^2_u$, while $\sigma^2_\emptyset=0$.
A larger $\sigma^2_u$ means a more important interaction
among those variables, but Sobol' indices
account for the fact that the importance of a set
of variables also depends on other interactions in
which they participate.

The two most important Sobol' indices are
\begin{align}
\ult^2_u = \sum_{v\subseteq u}\sigma^2_v,\quad\text{and}\\
\olt^2_u = \sum_{v\cap u\ne\emptyset}\sigma^2_v.
\end{align}
These satisfy $0\le\ult^2_u \le\olt^2_u\le\sigma^2$
and $\ult^2_u =\sigma^2-\olt^2_{-u}$, where $\sigma^2$
is the variance $\int (f(\bsx)-\mu)^2\rd\bsx$.
We use $-u$ or $u^c$ depending on typographical
readability, to denote the complement of $u$ in $\cd$.
These indices provide two measures of the importance
of the variables in subset $u$. The larger measure
includes interactions between variables in
$u$ and variables in its complement, while the smaller
measure excludes those interactions.

If we unite the part of $\bsx\in[0,1]^d$ corresponding
to indices in the set $u$ with the part of another
point $\bsy\in[0,1]^d$ for indices in $-u$, then the resulting point is
denoted $\bsx_u\glu\bsy_{-u}$.
Most estimation strategies for Sobol' indices are based on
the identities
\begin{align}
\ult^2_u&=\mu^2+\int f(\bsx)f(\bsx_u\glu\bsy_{-u})\rd\bsx\rd\bsy,\quad\text{and}\label{eq:basiclo}\\
\olt^2_u&=\frac12\int \bigl( f(\bsx)-f(\bsx_u\glu\bsy_{-u}))^2\rd\bsx\rd\bsy,\label{eq:basichi}
\end{align}
with integrals taken over $\bsx$ and $\bsy\in[0,1]^d$.

When $\olt^2_u$ is small, then~\eqref{eq:basichi} leads to a very
effective Monte Carlo strategy based on
$$
\wh\olt^2_u=\frac1{2n}\sum_{i=1}^n \bigl( f(\bsx_i)-f(\bsx_{i,u}\glu\bsy_{i,-u}))^2
$$
for $(\bsx_i,\bsz_i)\simiid\dustd[0,1]^d$.
This estimator is a sum of squares, hence nonnegative, and it is unbiased.
If the true $\ult^2_u=0$, then $\wh\ult^2_u=0$ with probability one.
More generally, if the true $\ult^2_u$ is small, then the
estimator averages squares of typically small quantities.
We assume throughout that $\int f(\bsx)^4\rd\bsx<\infty$
so that the variance of this and our other estimators is finite.

The natural way to estimate $\ult^2_u$ is via
\begin{align}\label{eq:natu}
\wh\ult^2_u=\frac1{n}\sum_{i=1}^n f(\bsx_i)f(\bsx_{i,u}\glu\bsy_{i,-u})
-\hat\mu^2.
\end{align}
The simplest estimator of $\hat \mu$ is
$(1/n)\sum_{i=1}^nf(\bsx_i)$ but
\cite{jano:klei:lagn:node:prie:2012:tr}
have recently proved that it is better to use
$\hat\mu = (1/2n)\sum_{i=1}^n
(f(\bsx_i)+f(\bsx_{i,u}\glu\bsy_{i,-u}))$.

\section{The estimators}\label{sec:esti}

The problem with~\eqref{eq:natu} is that
it has very large variance
when $\ult^2_u \ll \mu^2$. 
Although $\ult^2_u$ is invariant with
respect to shifts replacing $f(\bsx)$
by $f(\bsx)-c$ for a constant $c$,
the variance of~\eqref{eq:natu} can
be strongly affected by such shifts.
\cite{sobo:1990,sobo:1993}
recommends shifting $f$ by an amount
close to $\mu$, which while not necessarily
optimal, should be reasonable.

An approximation to $\mu$ can be obtained
by Monte Carlo or quasi-Monte Carlo simulation
prior to estimation of $\ult_u^2$. In our simulations
we suppose that an oracle has supplied $\mu$
and then we compare estimators that do and
do not benefit from the oracle.

Another estimator of $\ult^2_u$ 
was considered independently by~\cite{salt:2002}
and the Masters thesis of~\cite{maun:2002} under the supervision
of S.\ S.\ Kucherenko and C.\ Pantelides.
This estimator, called correlated sampling by~\cite{sobo:mysh:2007}
replaces $f(\bsx_{i,u}\glu\bsy_{i,u})$
by $f(\bsx_{i,u}\glu\bsy_{i,u})-f(\bsy)$
in~\eqref{eq:natu} and then it is no longer
necessary to subtract $\hat\mu^2$.
Indeed the method can be viewed as subtracting the estimate
$n^{-2}\sum_{i=1}^n\sum_{i'=1}^nf(\bsx_i)f(\bsy_{i'})$
from the sample mean of $f(\bsx)f(\bsx_u\glu\bsy_{-u})$.
That estimator is called ``Correlation 1'' below.

\cite{sobo:mysh:2007} find that even the
correlated sampling method has increased
variance when $\mu$ is large. They propose
another estimator  replacing the first $f(\bsx_i)$
by $f(\bsx_i)-c$ for a constant $c$ near $\mu$.
Supposing that an oracle has supplied $c=\mu$
we call the resulting method ``Oracle 1''
because it makes use of the true $\mu$ one time.
One could also make use of the oracle's $\mu$
in both the left and right members of
the cross moment pair. We call this estimator
``Oracle 2'' below.
The fourth method to compare is a new estimator,
called ``Correlation 2'', that uses two
random offsets. Instead of replacing
$f(\bsx_i)$ by $f(\bsx_i)-\mu$ it
draws a third variable $\bsz\sim\dustd[0,1]^d$
and is based on  the identity
\begin{align}\label{eq:3fold}
&\iiint \bigl(f(\bsx)-f(\bsz_u\glu\bsx_{-u})\bigr)
\bigl(f(\bsx_u\glu\bsy_{-u})-f(\bsy)\bigr)
\rd\bsx\rd\bsy\rd\bsz\notag\\
& =
(\mu^2+\ult^2_u) -\mu^2-\mu^2+\mu^2 - \ult^2_u.
\end{align}

We compare the following estimators
\begin{align*}
&\frac1n\sum_{i=1}^nf(\bsx_i)(f(\bsx_{i,u}\glu\bsy_{i,-u})-f(\bsy))&& \text{(Correlation 1)}\\
&\frac1n\sum_{i=1}^n(f(\bsx_i)-f(\bsz_{i,u}\glu\bsx_{i,-u})(f(\bsx_{i,u}\glu\bsy_{i,-u})-f(\bsy))&& \text{(Correlation 2)}\\
&\frac1n\sum_{i=1}^n(f(\bsx_i)-\mu)(f(\bsx_{i,u}\glu\bsy_{i,-u})-f(\bsy))&& \text{(Oracle 1)}\\
&\frac1n\sum_{i=1}^n(f(\bsx_i)-\mu)(f(\bsx_{i,u}\glu\bsy_{i,-u})-\mu)&& \text{(Oracle 2)}
\end{align*}
where 
$(\bsx_i,\bsy_i,\bsz_i)\simiid\dustd[0,1]^{3d}$ for $i=1,\dots,n$.
Not all components of these vectors are
necessary to estimate $\ult^2_u$ for a single
$u$, but many applications seek $\ult^2_u$ for several
sets $u$ at once, so it is simpler to write them this way.
Also, the cost is assumed to be largely in evaluating
$f$, not in producing the inputs $(\bsx_i,\bsy_i,\bsz_i)$.
The properties of these estimators are
given in Table~\ref{tab:est}.

The intuitive reason why ``Correlation 2'' should
be effective on small indices is as follows.
If the variables in the set $u$ are really unimportant
then $f(\bsx)$ will be determined almost entirely
by the values in $\bsx_{-u}$.
Then both $f(\bsx_i)-f(\bsz_{i,u}\glu\bsx_{i,-u})$
and $f(\bsx_{i,u}\glu\bsy_{i,-u})-f(\bsy)$
should be small values, even smaller than by
centering at $\mu$, and so the estimator
takes a mean product of small quantities.

We do not compare the original estimator~\eqref{eq:natu}.
The bias correction makes it more complicated to describe
the accuracy of this method.  Also that estimator
had extremely bad performance in~\cite{sobo:mysh:2007}.

\begin{table}\centering
\begin{tabular}{lcc}
\toprule
Name & Expectation & Cost \\
\midrule
Original~\eqref{eq:natu} & $\mu^2+\ult^2_u$ & 2 \\
Correlation 1 & $\ult^2_u$ & 3 \\
Correlation 2 & $\ult^2_u$ & 4 \\
Oracle 1 & $\ult^2_u$ & 3 \\
Oracle 2 & $\ult^2_u$ & 2 \\
\bottomrule
\end{tabular}
\caption{\label{tab:est}
Estimators of $\ult^2_u$ with expected value and number
of function values required per sample.
}
\end{table}

\section{Examples}\label{sec:exam}

\subsection{$g$ function}

This is the example used by~\cite{sobo:mysh:2007}.
It has $d=3$ and 
$$
f(\bsx) = \prod_{j=1}^3\frac{|4x_j-2|+2+3a}{1+a_j}.
$$
This function has $\mu=27$ and
$\sigma^2_{\{1\}} = 0.0675$,
$\sigma^2_{\{2\}} = 0.27$,
$\sigma^2_{\{3\}} = 1.08$,
$\sigma^2_{\{1,2\}} = 0.000025$,
$\sigma^2_{\{1,3\}} = 0.0001$,
$\sigma^2_{\{2,3\}} = 0.0004$,
$\sigma^2_{\{2,3\}} = 0.0004$,
and $\sigma^2_{\{1,2,3\}} \doteq 3.7\times 10^{-8}$.
The smallest and therefore probably
the most difficult $\ult^2_u$ to estimate
is $\ult^2_{\{1\}}=\sigma^2_{\{1\}}$.
That is the one that they measure.

They report numerical values of $\wh\ult^2_{\{1\}}/\ult^2_{\{1\}}$
for the four estimates in Table~\ref{tab:est}
(exclusive of the new ``Correlation 2'' estimate)
based on $n=256{,}000$ samples. The original estimator
gave a values $2.239$ times as large as
the true $\ult^2_{\{1\}}$. The others were 
ranged from $0.975$ to $1.104$ times the true value.
They did not use the oracle for $\mu$, but centered
their estimator instead on $c=26.8$ to investigate
a somewhat imperfect oracle.

The four estimators we consider here are all simply
sample averages. As a result we can measure their
efficiency by just estimating their variances.
The efficiencies of these methods, using 
``Correlation 1'' as the baseline are given by
$$
E_{\text{corr 2}} = 
\frac34
\frac{\var(\text{corr 1})}{\var(\text{corr 2})},\quad
E_{\text{orcl 1}} = 
\frac{\var(\text{corr 1})}{\var(\text{orcl 1})},\quad\text{and}\quad
E_{\text{orcl 2}} = 
\frac32
\frac{\var(\text{corr 1})}{\var(\text{orcl 2})}
$$
where the multiplicative factors accounts for the 
unequal numbers
of function calls required by the methods.

\begin{table}
\centering
\begin{tabular}{cccrrr}
\toprule
Set $u$ &  $\ult_u^2/\sigma^2$ & Corr 1 & Corr 2 & Orcl 1  & Orcl 2 \\
\midrule
$\{1\}$& $0.048$ &1 & 4256 & 518 &  74 \\
$\{2\}$& $0.190$ &1 & 1065 & 525 &  297 \\
$\{3\}$& $0.762$ &1 & 267 & 556 & 1329 \\
$\{1,2\}$& $0.238$ &1 & 774 & 503 &  364 \\
$\{1,3\}$& $0.809$ &1 & 243 & 529 & 1306 \\
$\{2,3\}$& $0.952$ &1 & 194 & 473 & 1261 \\
\bottomrule
\end{tabular}
\caption{\label{tab:effg}
Relative efficiencies of $4$ estimators of $\ult^2_u$
for the $g$-function, rounded to the nearest integer.
Relative indices $\ult^2_u/\sigma^2$ rounded to three places.
}
\end{table}

The efficiencies of the four estimators
are compared in Table~\ref{tab:effg} based
on $n=1{,}000{,}000$ function evaluations. This is far
more than one would ordinarily
use to estimate the indices themselves, but
we are interested in their sampling variances
here.
We consider all sets except $u=\{1,2,3\}$
because $\ult^2_{\{1,2,3\}}=\sigma^2$ which
can be estimated more directly.
The table contains one small index
$\ult^2_{\{1\}}$, (the
one \cite{sobo:mysh:2007} studied).
On the small effect, the new Correlation 2 estimator
is by far the most efficient, outperforming both
oracles. Inspecting the table, it is clear
that it pays to use subtraction in both
left and right sides of the estimator
and that the smaller the effect $\ult^2_u$ is, the
better it is to replace the oracle's $\mu$ with
a correlation based estimate.

\subsection{Other product functions}

It is convenient to work with functions of the form
\begin{align}
\label{eq:prod}
f(\bsx) = \prod_{j=1}^d (\mu_j + \tau_j g_j(x_j))
\end{align}
where $\int_0^1g(x)\rd x=0$,
 $\int_0^1g(x)^2\rd x=1$, and
 $\int_0^1g(x)^4\rd x<\infty$.
For this function $\sigma^2_u = \prod_{j\in u}\tau_j^2\prod_{j\not\in u}\mu^2_j$.
Taking $g(x) = \sqrt{12}(x-1/2)$,
$d=6$,
$\mu=(1,1,1,1,1,1)$
and $\tau=(4,4,2,2,1,1)/4$
and sampling $n=1{,}000{,}000$ times lead
to the results in Table~\ref{tab:effg1}.
The results are not as dramatic
as for the $g$-function, but they show
the same trends.  The smaller $\ult^2_u$ is,
the more improvement comes from the
new estimator. On the smallest indices
it beats both oracles.

\begin{table}[t]
\centering
\begin{tabular}{cccrrr}
\toprule
Set $u$ &  $\ult_u^2/\sigma^2$ & Corr 1 & Corr 2 & Orcl 1  & Orcl 2 \\
\midrule
$\{1\}$& $0.165$ &1 & $0.74$ & $1.13$ & $1.23$ \\
$\{2\}$& $0.165$ &1 & $0.73$ & $1.14$ & $1.24$ \\
$\{3\}$& $0.041$ &1 & $1.69$ & $1.15$ & $0.54$ \\
$\{4\}$& $0.041$ &1 & $1.67$ & $1.15$ & $0.54$ \\
$\{5\}$& $0.010$ &1 & $5.45$ & $1.16$ & $0.20$ \\
$\{6\}$& $0.010$ &1 & $5.58$ & $1.16$ & $0.20$ \\
\midrule        
$\{1,2\}$& $0.826$ & 1 & $0.75$ & $1.21$ & $1.86$\\
$\{3,4\}$& $0.176$ & 1 & $1.23$ & $1.16$ & $0.94$\\
$\{5,6\}$& $0.042$ & 1 & $2.94$ & $1.16$ & $0.38$\\
\bottomrule
\end{tabular}
\caption{\label{tab:effg1}
Relative efficiencies of $4$ estimators of $\ult^2_u$
for the product function~\eqref{eq:prod}.
Relative indices $\ult^2_u/\sigma^2$ rounded to three places.
}
\end{table}

The improvements for the $g$-function are
much larger than for the product studied
here. For the purposes of Monte Carlo sampling
the absolute value cusp in the $g$-function
makes no difference. The $g$-function has
the same moments as the product function
with $\mu_j = 3$ and $\tau_j = 1/(\sqrt{3}a_j)$.
Computing the $g$ function estimates with
the product function code (as a check) yields the same
magnitude of improvement seen in Table~\ref{tab:effg}.


\section{Some generalizations and  a recommendation}\label{sec:gene}

The best unbiased estimator of $\ult^2_u$ is the one that minimizes
the variance after making an adjustment for
the number of function calls.
Unfortunately variances of estimated variances involve fourth
moments which are harder to ascertain than the second
moments underlying the ANOVA decomposition.

\subsection{More general centering}
The estimators in Section~\ref{sec:exam} are all formed by
taking pairs $f(\bsx)$ and $f(\bsx_u\glu\bsy_{-u})$,
subtracting centers from them, and averaging the product
of those two centered values.
Where they differ is in how they are centered.

We can generalize this approach to a spectrum of 
centering methods.

\begin{theorem}\label{thm:four}
Let $v$ and $v'$ be two subsets of $u^c$
and let $\bsx,\bsy,\bsw,\bsz$ be independent
$\dustd[0,1]^d$ random vectors.
Then
\begin{align}\label{eq:four}
\e\Bigl(
\bigl(
f(\bsx)-f(\bsx_v\glu\bsz_{-v})
\bigr)
\bigl(
f(\bsx_u\glu\bsy_{-u})-f(\bsy_{v'}\glu\bsw_{-v'})
\bigr)
\Bigr)
&=\ult^2_u.
\end{align}
\end{theorem}
\begin{proof}
\begin{align*}
&\e\Bigl(
\bigl(
f(\bsx)-f(\bsx_v\glu\bsz_{-v})
\bigr)
\bigl(
f(\bsx_u\glu\bsy_{-u})-f(\bsy_{v'}\glu\bsw_{-v'})
\bigr)
\Bigr)\notag\\
&=
(\mu^2 + \ult^2_u)
-(\mu^2 + \ult^2_\emptyset)
-(\mu^2 + \ult^2_{u\cap v})
+(\mu^2 + \ult^2_\emptyset)\notag\\
&=\ult^2_u,
\end{align*}
because $u\cap v=\emptyset$ and $\ult^2_\emptyset=0$.
\end{proof}

As a result of Theorem~\ref{thm:four},
we may estimate $\ult_u^2$ by
\begin{align}\label{eq:genult}
\frac1n\sum_{i=1}^n
\bigl(
f(\bsx_i)-f(\bsx_{i,v}\glu\bsz_{i,-v})
\bigr)
\bigl(
f(\bsx_{i,u}\glu\bsy_{i,-u})-f(\bsy_{i,v'}\glu\bsw_{i,-v'})
\bigr)
\end{align}
where 
$(\bsx_i,\bsy_i,\bsw_i,\bsz_i)\simiid\dustd(0,1)^{4d}$.

The new estimate~\eqref{eq:genult} uses up four
independent vectors, not the three used in
the Correlation 2 estimator, so we should check
that it really is a generalization.

First, suppose that $v'=u^c$. Then the only part of
the vector $\bsw$ that is used in~\eqref{eq:genult}
is $\bsw_{-v'}=\bsw_u$.
Because~\eqref{eq:genult} does not use $\bsy_u$
the needed parts of $\bsy$ and $\bsw$ fit
within the same vector. That is
we can sample $\bsy$ as before and use
$\bsy_u$ for $\bsw_u$. As a result when
$v'=u^c$ we only need three vectors as follows:
\begin{align}\label{eq:genult2}
\frac1n\sum_{i=1}^n
\bigl(
f(\bsx_i)-f(\bsx_{i,v}\glu\bsz_{i,-v})
\bigr)
\bigl(
f(\bsx_{i,u}\glu\bsy_{i,-u})-f(\bsy)
\bigr).
\end{align}
If we take $v=u^c$ too, then~\eqref{eq:genult2}
reduces to the Correlation 2 estimator.

At first, it might appear that the Oracle 2
estimator arises by taking $v=v'=\emptyset$,
but this is not what happens, even when $\mu=0$.
A more appropriate generalization of the oracle
estimators is to based on the identity
$$
\ult^2_v = \e\bigl(
\bigl(f(\bsx)-\mu_v(\bsx_v)\bigr)
\bigl(f(\bsx_u\glu\bsz_{-u})-\mu_{v'}(\bsz_{v'})\bigr)
\bigr)$$
where $\mu_v(\bsx_v) = \e(f(\bsx)\mid \bsx_v)$
and $v,v'\subseteq u^c$.
To turn this identity into a practical estimator
requires estimation of these conditional expectations.
For $v=v'=\emptyset$ the conditional expectations
become the unconditional expectation, which is simply
the integral of $f$. For other $v$ and $v'$, such
estimation requires something like nonparametric
regression, with bias and variance expressions
that complicate the analysis of the resulting
estimate.

\subsection{Recommendation}

The Correlation 2 estimator has $v=v'=u^c$,
so it holds constant all of the variables
in $\bsx_{-u}$.
From Theorem~\ref{thm:four}, we see that this
is just one choice among many and it
raises the question of which variables
should be held fixed in a Monte Carlo estimate
of $\ult^2_u$. 
The result is that we find taking $v=v'=u^c$ to be a
principled choice.

We can get some insight by considering functions
of product form.  Even there the resulting variance
formulas become cumbersome, but simplified versions
yield some insight. 
We can write it as
\begin{align}\label{eq:prodwithh}
f(\bsx) = \prod_{j=1}^d h_j(x_j)
\end{align}
where $h_j(x) = \mu_j +\tau_jg_j(x)$
with
$\int_0^1g_j(x)^p\rd x$ taking values
$0$, $1$, $\gamma_j$ and $\kappa_j$
for $p=1$, $2$, $3$, and $4$ respectively.
In statistical terms, the random variable
$h_j(x)$ has skewness $\gamma_j/\tau_j^3$
and kurtosis $\kappa_j/\tau_j^4-3$
if $x\sim\dustd[0,1]$ and $\tau_j>0$.
We will suppose that all $\tau_j\ge 0$
and that all $\kappa_j<\infty$.

\begin{proposition}\label{prop:4moment}
Let $\wh\ult^2_u$ be given by~\eqref{eq:genult},
where $f$ is given by the product model~\eqref{eq:prodwithh}.
Then, for $v,v'\subseteq u^c$,
$$
n\var(\wh\ult^2_u)
=\e\bigl( Q_{v}(\bsx,\bsy,\bsz,\bsw)
 Q_{uv'}(\bsx,\bsy,\bsz,\bsw)
\bigr) - \ult^4_u
$$
for $\bsx,\bsy,\bsz,\bsw\simiid\dustd[0,1]^d$
where
\begin{align*}
Q_v 
&= \prod_{j=1}^dh^2_j(x_j)
+\prod_{j\in v}h^2_j(x_j)\prod_{j\not\in v}h_j^2(z_j)
-2\prod_{j\in v}h_j^2(x_j)\prod_{j\not\in v}h_j(x_j)h_j(z_j),\quad\text{and}\\
Q_{uv'} 
&= \prod_{j\in u}h^2_j(x_j)\prod_{j\not\in u}h^2_j(y_j)
+\prod_{j\in v'}h^2_j(y_j)\prod_{j\not\in u}h^2_j(w_j)\\
&\phantom{=}-2
\prod_{j\in u^c\cap v'}
h^2_j(y_j)
\prod_{j\in u\cap v'^c}
h_j(x_j)h_j(w_j)
\prod_{j\in u^c\cap v'^c}
h_j(y_j)h_j(w_j)
\end{align*}
\end{proposition}
\begin{proof}
We need the expected square of the quantity
inside the expectation in equation~\eqref{eq:four}.
First we expand
\begin{align*}
f(\bsx)-f(\bsx_{v}\glu\bsz_{-v})
=\prod_{j=1}^dh_j(x_j)-\prod_{j\in v}h_j(x_j)\prod_{j\not\in v}h_j(z_j).
\end{align*}
Squaring this term yields $Q_v$, and similarly, squaring
\begin{align*}
f(\bsx_u\glu\bsy_{-u})-f(\bsy_{v'}\glu\bsw_{-v'})
=\prod_{j\in u}h_j(x_j)\prod_{j\not\in u}h_j(y_j)
-\prod_{j\in v'}h_j(y_j)\prod_{j\not\in v'}h_j(w_j)
\end{align*}
yields $Q_{uv'}$, after using $u\cap v'=\emptyset$.
\end{proof}

Using Proposition~\ref{prop:4moment} we can see
what makes for a good estimator in the product
function setting. The quantities $Q_v$ and
$Q_{uv'}$ should both have small variance
and their correlation should be small.
The latter effect is very complicated
depending on the interplay among
$u$, $v$ and $v'$,
and one might expect it to be of lesser
importance. So we look at $\e( Q_v^2)$
for insight as to which indices should
be in $v$.  Then we suppose that it
will usually be best to take the same
indices for both $v$ and $v'$.

\begin{theorem}\label{thm:EQsq}
Let $\wh\ult^2_u$ be given by~\eqref{eq:genult},
where $f$ is given by the product model~\eqref{eq:prodwithh}
and let $Q_v$ be as defined in
Proposition~\eqref{prop:4moment}.
Then $Q_v$ is minimized over $v\subseteq u^c$ by
taking $v=u^c$.
\end{theorem}
\begin{proof}
Let $\mu_{4j} = \int_0^1 h_j(x)^4\rd x$
and $\mu_{2j} = \int_0^1 h_j(x)^2\rd x$.
It is elementary that $\mu_{4j}\ge\mu_{2j}^2$.
Expanding $\e(Q_u^2)$ and gathering terms yields,
\begin{align*}
&\prod_{j=1}^d\mu_{4j}
+\prod_{j=1}^d\mu_{4j}
+4\prod_{j\in v}\mu_{4j}\prod_{j\not\in v}\mu_{2j}^2
+2\prod_{j\in v}\mu_{4j}\prod_{j\not\in v}\mu_{2j}^2\\
&-4\prod_{j\in v}\mu_{4j}\prod_{j\not\in v}\mu_{2j}^2
-4\prod_{j\in v}\mu_{4j}\prod_{j\not\in v}\mu_{2j}^2\\
& =
2\prod_{j=1}^d\mu_{4j}-
2\prod_{j\in v}\mu_{4j}\prod_{j\not\in v}\mu_{2j}^2.
\end{align*}
We minimize this expression by taking the
largest possible set $v\subseteq u^c$,
that is $v=u^c$.
\end{proof}

From Theorem~\ref{thm:EQsq}, we see that the
Correlation 2 estimator minimizes $\e(Q^2_u)$
for product functions among estimators
of the form~\eqref{eq:genult}.

\section*{Acknowledgments}

This work was supported by the U.S.\ National
Science Foundation under grant DMS-0906056.
I thank Sergei Kucherenko for helpful discussions.

\bibliographystyle{apalike}
\bibliography{sensitivity}

\begin{thebibliography}{}

\bibitem[Janon et~al., 2012]{jano:klei:lagn:node:prie:2012:tr}
Janon, A., Klein, T., Lagnoux, A., Nodet, M., and Prieur, C. (2012).
\newblock {Asymptotic normality and efficiency of two Sobol' index estimators}.
\newblock Technical report, INRIA.

\bibitem[Mauntz, 2002]{maun:2002}
Mauntz, W. (2002).
\newblock Global sensitivity analysis of general nonlinear systems.
\newblock Master's thesis, Imperial College.
\newblock Supervisors: C. Pantelides and S. Kucherenko.

\bibitem[Owen, 2012]{sobomat:tr}
Owen, A.~B. (2012).
\newblock Variance components and generalized {Sobol'} indices.
\newblock Technical report, Stanford University.

\bibitem[Saltelli, 2002]{salt:2002}
Saltelli, A. (2002).
\newblock Making best use of model evaluations to compute sensitivity indices.
\newblock {\em Computer Physics Communications}, 145:280--297.

\bibitem[Saltelli et~al., 2008]{salt:ratt:andr:camp:cari:gate:sais:tara:2008}
Saltelli, A., Ratto, M., Andres, T., Campolongo, F., Cariboni, J., Gatelli, D.,
  Saisana, M., and Tarantola, S. (2008).
\newblock {\em Global Sensitivity Analysis. The Primer}.
\newblock John Wiley \& Sons, Ltd, New York.

\bibitem[Sobol', 1990]{sobo:1990}
Sobol', I.~M. (1990).
\newblock On sensitivity estimation for nonlinear mathematical models.
\newblock {\em Matematicheskoe Modelirovanie}, 2(1):112--118.
\newblock (In Russian).

\bibitem[Sobol', 1993]{sobo:1993}
Sobol', I.~M. (1993).
\newblock Sensitivity estimates for nonlinear mathematical models.
\newblock {\em Mathematical Modeling and Computational Experiment}, 1:407--414.

\bibitem[{Sobol'} and Myshetskaya, 2007]{sobo:mysh:2007}
{Sobol'}, I.~M. and Myshetskaya, E.~E. (2007).
\newblock {Monte Carlo} estimators for small sensitivity indices.
\newblock {\em {Monte Carlo} methods and their applications},
  13(5--6):455--465.

\end{thebibliography}

\end{document}